\documentclass[envcountsect,envcountsame,runningheads,pdftex,a4paper]{llncs}

\usepackage{FM2019}

\allowdisplaybreaks

\begin{document}

\title{Mechanised Assurance Cases with \\ Integrated Formal Methods in Isabelle}
\titlerunning{Assurance Cases in Isabelle}

\author{Yakoub Nemouchi \and Simon Foster \and Mario Gleirscher \and Tim Kelly}
\institute{University of York \\ \email{firstname.lastname@york.ac.uk}}

\authorrunning{Yakoub Nemouchi \and Simon Foster \and Mario Gleirscher \and Tim Kelly}

\maketitle

\begin{abstract}
  Assurance cases are often required as a means to certify a critical system. Use of formal methods in assurance can
  improve automation, and overcome problems with ambiguity, faulty reasoning, and inadequate evidentiary
  support. However, assurance cases can rarely be fully formalised, as the use of formal methods is contingent on models
  validated by informal processes. Consequently, we need assurance techniques that support both formal and informal
  artifacts, with explicated inferential links and assumptions that can be checked by evaluation. Our contribution is a
  mechanical framework for developing assurance cases with integrated formal methods based in the Isabelle system. We
  demonstrate an embedding of the Structured Assurance Case Meta-model (SACM) using Isabelle/DOF, and show how this can
  be linked to formal analysis techniques originating from our verification framework, Isabelle/UTP. We validate our
  approach by mechanising a fragment of the Tokeneer security case, with evidence supplied by formal verification.
\end{abstract}


%%%%%%%%%%%%%%%%%%%%%%
\section{Introduction}
%%%%%%%%%%%%%%%%%%%%%%
\label{sec:intro}
Cyber-physical systems (CPS) control critical socio-technical processes prone to faults and other critical events with
potentially undesired consequences. Such systems include autonomous vehicles, traffic flow control, patient monitoring,
surgical robot assistants, and building security automation.  Real-time concurrency of physical events and computation
poses tough challenges in achieving high levels of assurance in verification and validation. Consequently, the benefits
of CPS can only be harnessed if they acquire consumer trust and regulatory acceptance.

Safety cases~\cite{Kelly1998,Hawkins2015}, and more generally assurance cases, are structured arguments, supported by
evidence, intended to convice a regulator that a system is acceptably safe for application in a specific operating
environment~\cite{Habli2010}. They are recommended by several international standards, such as ISO26262 for automotive
applications. An assurance case consists of a hierarchical decomposition of requirements, through appropriate
argumentation strategies, into further claims, and eventually supporting evidence. Several languages exist for
expressing assurance cases, including the Goal Structuring Notation~\cite{Kelly1998} (GSN), and the closely related
Structured Assurance Case Metamodel\footnote{\textit{Structured Assurance Case Metamodel}:
  \url{http://www.omg.org/spec/SACM/}} (SACM).

Assurance case creation can be supported by model-based design, which utilises architectural and behavioural models over
which requirements can be formulated~\cite{Habli2010}. However, safety cases can suffer from undermining logical
fallacies and lack of evidence~\cite{Greenwell2006}. A proposed solution is formalisation in a machine-checked logic to
enable verification of consistency and well-foundedness~\cite{Rushby2013}. As confirmed by avionics standard DO-178C
supplement DO-333, the evidence gathering process can also benefit from the rigour of formal methods. At the same time,
we acknowledge that, (1) assurance cases are intended primarily for human consumption, and (2) that formal models must
be validated informally~\cite{Habli2014}. Consequently, assurance cases will usually combine informal and formal
content, and any tool must support this.

Our vision is a unified framework for machine-checked assurance cases, and with evidence provided by a number of
integrated formal methods~\cite{Gleirscher2018-NewOpportunitiesIntegrated}. Such a framework can have a transformative
effect in the field of assurance by harnessing results from automated formal verification to produce assurance cases
undergirded by specific mathematical guarantees of their consistency and adequacy of the evidence. Moreover, it can
provide a potential route to regulatory acceptance, through the production of mathematically verified safety
certificates.

The contributions of this paper make a first step in this direction: (1) an implementation of SACM in the Isabelle
interactive theorem prover~\cite{Nipkow-Paulson-Wenzel:2002}, (2) a machine-checked domain-specific assurance language,
and (3) integration of formal evidence from our verification framework,
Isabelle/UTP~\cite{Foster19a-IsabelleUTP}. Isabelle provides a sophisticated executable document model for presenting a
graph of hyperlinked artifacts, like definitions, theorems, and proofs. The document model provides automatic and
incremental consistency checking, and change analysis, where updates to model artifacts trigger rechecking. Such
capabilities can support efficient maintenance and evolution of model-based assurance cases~\cite{Hawkins2015}.

Moreover, the document model allows management of both informal and formal content, and access to a vast array of
automated verification tools~\cite{Wenzel2007}. In particular, our own verification framework,
Isabelle/UTP~\cite{Foster19a-IsabelleUTP,Foster16a}, harnesses Hoare and He's Unifying Theories of
Programming~\cite{Hoare&98} (UTP) to provide verification facilities for a variety of programming and modelling
languages with paradigms as diverse as concurrency, real-time, and hybrid computation. We validate our approach by
mechanising an assurance case for the Tokeneer system~\cite{Barnes2006-EngineeringTokeneerenclave}, including the
underlying formal model and verification of security functional requirements\footnote{Supporting materials, including
  Isabelle theories, can be found on \href{https://www.cs.york.ac.uk/~simonf/FM2019/}{our website}}.

In \S\ref{sec:prelim} we outline preliminary materials: SACM, Isabelle, and the Isabelle/DOF ontology framework. In
\S\ref{sec:tokeneer} we describe the Tokeneer system, and how it is assured and verified. In \S\ref{subsec:isacm} we
begin our contributions by describing the implementation of \isacm and our assurance DSL in Isabelle. In
\S\ref{sec:model} we describe how we model and verify Tokeneer using our verification framework, Isabelle/UTP. In
\S\ref{sec:tokassure} we describe the mechanisation of the assurance case for Tokeneer in \isacm. In
\S\ref{subsec:relatedWork} we highlight related work, and in \S\ref{sec:conclusion} we conclude.
%%%%%%%%%%%%%%%%%%%%%%%
\section{Preliminaries}
%%%%%%%%%%%%%%%%%%%%%%
\label{sec:prelim}
\begin{wrapfigure}{r}{0.57\linewidth}
  \vspace{-5.5ex}

  \centering
  \includegraphics[width=7cm]{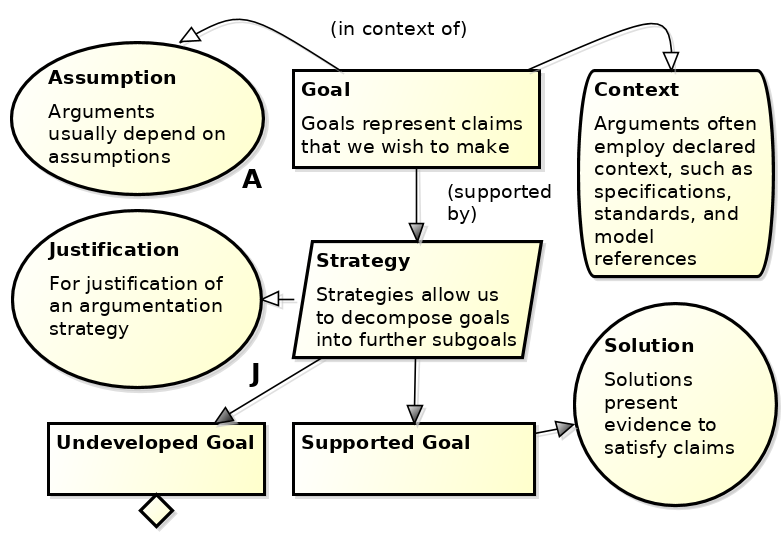}

  \vspace{-2.5ex}
  
  \caption{Goal Structuring Notation}
  
  \label{fig:gsn}

  \vspace{-3.5ex}

\end{wrapfigure}

\textbf{SACM.} Assurance cases are often presented using a notation like GSN~\cite{Kelly1998} (Figure~\ref{fig:gsn}), that shows the
claims that are made, the argumentation strategies, the contextual elements, assumptions, justifications, and eventually
evidence. SACM is an OMG standard meta-model for assurance cases~\cite{Hawkins2015,Selviandro2018-SACM}. It aims at
unifying and refining a variety of predecessor notations, including GSN~\cite{Kelly1998} and CAE (Claims, Arguments, and
Evidence), and is intended to be a definitive reference model.

SACM has three crucial concepts: arguments, artifacts, and terminology. An argument consists of a collection of claims,
evidence citations, and inferential links between them. Artifacts manifest evidence, such as models, techniques,
results, verification activities, and participants. Terminology is used to fix formal terms for the use in
claims. Normally, claims are textual, but in SACM they can also contain structured expressions, which allows integration
of formal languages.

\begin{figure}[t]

  \vspace{-1ex}

  \centering\includegraphics[width=.8\linewidth]{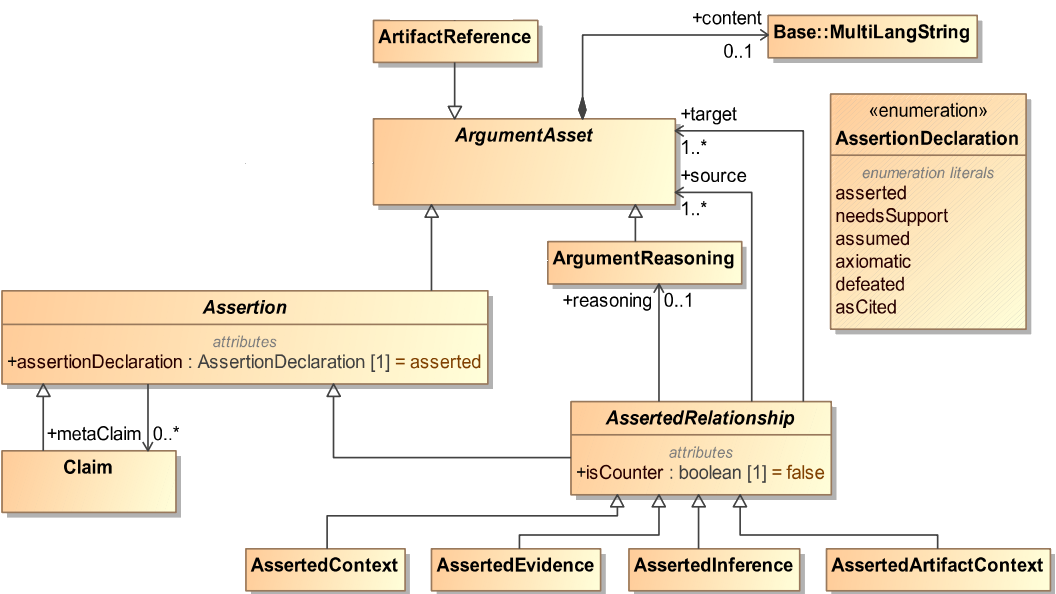}

  \caption{SACM Argumentation Meta-Model}
  \label{fig:sacm-arg}

  \vspace{-4ex}

\end{figure}

The argumentation meta-model is shown in Figure~\ref{fig:sacm-arg}. The base class is \textsf{ArgumentAsset}, which
groups the argument assets, such as \textsf{Claim}s, \textsf{ArtifactReference}s, and \textsf{AssertedRelationship}s
(which are inferential links). Every asset may contain a \textsf{MultiLangString} that provides a description,
potentially in multiple natural and formal languages, and corresponds to contents of the shapes in Figure~\ref{fig:gsn}.

\textsf{AssertedRelationship}s represent a relationship that exists between several assets. They can be of type
\textsf{AssertedContext}, which uses an artifact to define context; \textsf{AssertedEvidence}, which evidences a claim;
\textsf{AssertedInference} which describes explicit reasoning from premises to conclusion(s); or
\textsf{AssertedArtifactSupport} which documents an inferential dependency between the claims of two artifacts.

Both \textsf{Claim}s and \textsf{AssertedRelationship}s inherit from \textsf{Assertion}, because in SACM both claims and
inferential links are subject to argumentation and refutation. SACM allows six different classes of assertion, via the
attribute \textsf{assertionDeclaration}, including \textsf{axiomatic} (needing no further support), \textsf{assumed},
and \textsf{defeated}, where a claim is refuted. An \textsf{AssertedRelationship} can also be flagged as
\textsf{isCounter}, where counterevidence for a claim is presented.

\vspace{1ex}

\noindent\textbf{Isabelle.} \ihol is an interactive theorem prover for
\ac{hol}~\cite{Nipkow-Paulson-Wenzel:2002}, based on the generic framework \iisar~\cite{Wenzel02isabelle/isar}. The
former provides a functional specification language, and a large array of facilities for proof and automated
verification. The latter has an interactive, extensible, and executable document model, which describes Isabelle
theories. An Isabelle theory contains a sequence of executable markup commands with a semantics given in the
meta-language SML.

\begin{wrapfigure}{r}{0.4\linewidth}
  \vspace{-5ex}

  \centering\includegraphics[scale=0.35]{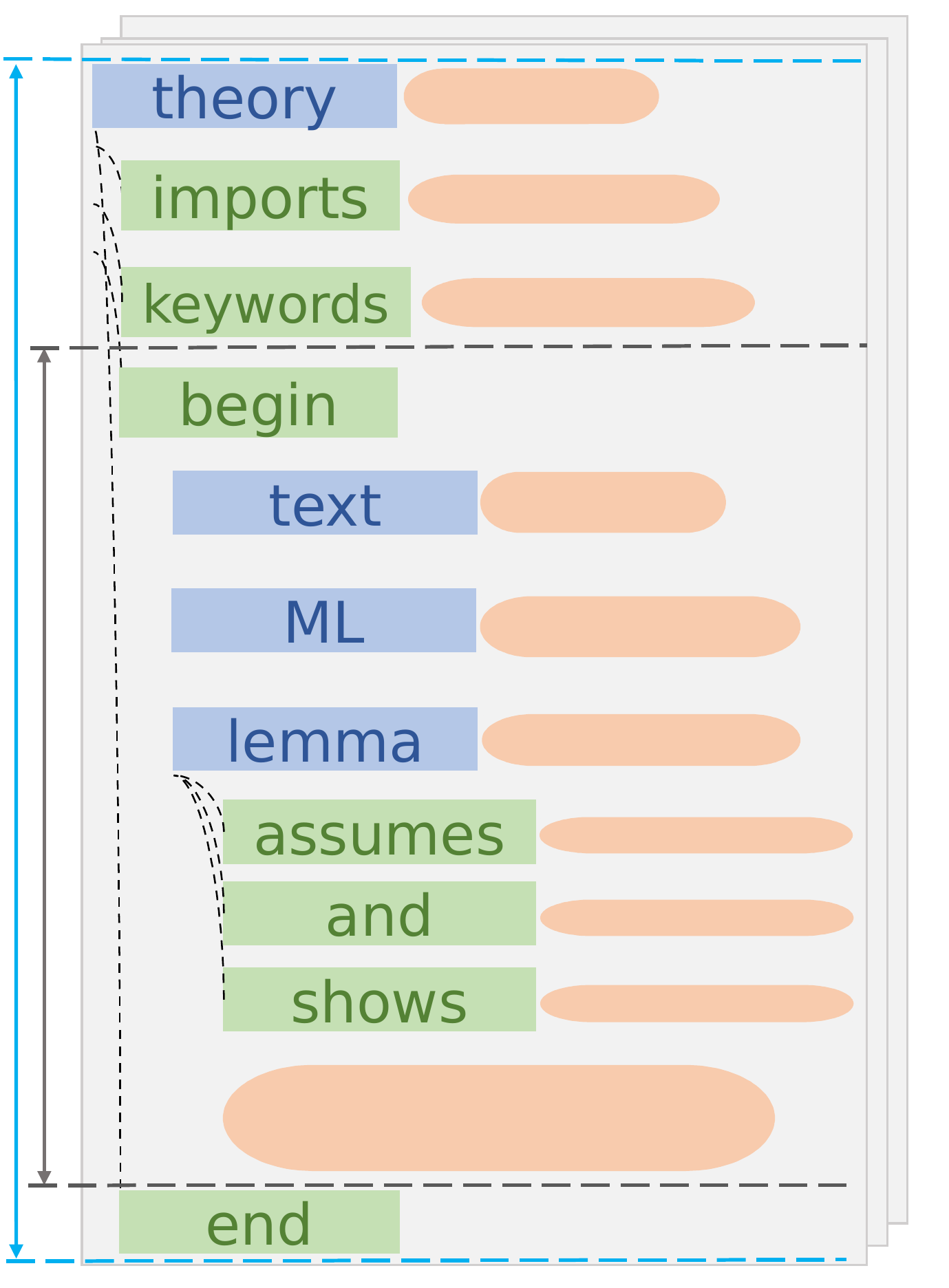}

  \vspace{-1ex}

  \caption{Document Model}\label{fig:doc-model}

  \vspace{-3ex}
\end{wrapfigure}

\autoref{fig:doc-model} gives an overview of the document model. The first section for \emph{context definition}
describes \emph{imports} of existing theories, and \emph{keywords} which give extensions to the concrete syntax. The
second section is the body enclosed between \emph{begin}-\emph{end} which is a sequence of commands. Isabelle commands
have a concrete syntax consisting of pre-declared top-level keywords (in \textcolor{Blue}{blue}), such as the command
\textcolor{Blue}{\inlineisar+ML+}, followed by a ``semantics area'' enclosed between
\inlineisar+\<open>...\<close>+. The keywords can be associated with optional attribute keywords (in
\textcolor{OliveGreen}{green}). The processing of the concrete syntax and any extensions is performed by SML code.

An Isabelle session is as an acyclic graph grouping several theories and their dependencies. When an edit is made to a
document in a session, it is immediately processed and executed, with feedback provided to the user. For example,
whenever the dependency structure of a document changes due to the removal, addition, or alteration of artifacts,
Isabelle reruns the associated code and any dependencies. This feature makes Isabelle ideal for assurance cases, which
have to be updated with every increment in system development

In addition to formal content, Isabelle theories can also contain informal commentary.  The
\textcolor{Blue}{\inlineisar+text+} \inlineisar+\<open>...\<close>+ command is a processor for textual markup content
containing a mixture of informal content, and links to formal document entities through \emph{antiquotations} of the
form \inlineisar+@{aqname ...}+. Antiquotations trigger a series of checks, for example the antiquotation
\inlineisar+@{thm+ \inlineisar+\<open>+\inlineisar+HOL.refl+\inlineisar+\<close>}+ checks if the theorem
\inlineisar+HOL.refl+ exists within the underlying theory context, and if so inserts a hyperlink.

Plugins, such as \ihol, HOL-TestGen~\cite{DBLP:journals/fac/BruckerW13}, and \idof~\cite{DBLP:conf/mkm/BruckerACW18}
contain document models and conservative extensions, following the LCF approach. \idof~\cite{DBLP:conf/mkm/BruckerACW18}
is a plugin that the \iisar document model implemented with support for ontologies. The result is a machine-checked
document model with formal hyperlinks between document instances of the modelled ontology.

The central component of \idof is the \ac{iosl}, which describes the content of documents in terms of several document
classes. Document classes can be linked to form a class model. We refer to~\cite{DBLP:conf/mkm/BruckerACW18} for
examples to model document content within \idof. A document class is the main entity in \ac{iosl} and it is represented
using the command \textcolor{Blue}{\inlineisar+doc_class+}, which creates a new class with a number of typed
attributes. The attributes of \textcolor{Blue}{\inlineisar+doc_class+} can refer both to the standard \ac{hol} types
such as \inlineisar+string+, \inlineisar+bool+, and also internal Isabelle meta-types such as \inlineisar+thm+,
\inlineisar+term+, or \inlineisar+typ+, which represent theorems, logical terms, and types, respectively. This is
because DOF ontologies sit a the meta-logical level, and so they can freely mix formal and informal content. This is our
motivation for its use in mechanising SACM.

%%%%%%%%%%%%%%%%%%%%%%%%%%%%%%%%%%%
\section{Running Example: Tokeneer}
%%%%%%%%%%%%%%%%%%%%%%%%%%%%%%%%%%%
\label{sec:tokeneer}
To demonstrate our approach, we use the \ac{tis}\footnote{Project
  website: \url{https://www.adacore.com/tokeneer}} illustrated in
Figure~\ref{fig:tis}, a system that guards entry to a secure
enclave by ensuring that only authorised users are admitted.

\begin{wrapfigure}{r}{0.5\linewidth}
  \vspace{-4.5ex}

  \centering
  \includegraphics[width=6cm]{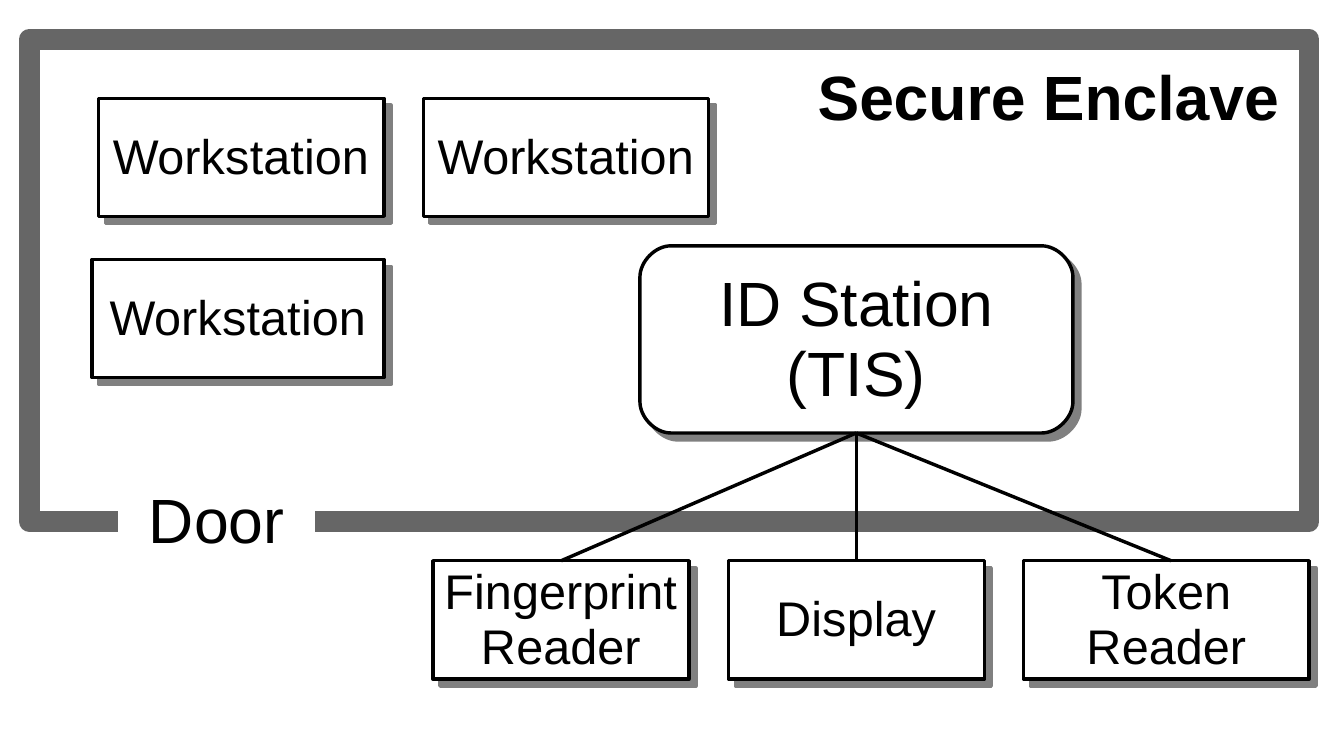}

  \vspace{-2.5ex}
  
  \caption{Tokeneer System Overview}
  
  \label{fig:tis}

  \vspace{-4ex}
\end{wrapfigure}
The relevant physical infrastructure consists of a door, a fingerprint reader, a display, and a card (token) reader. The
main function of the TIS is to check the stored credentials on a presented token, read a fingerprint if necessary, and
then either unlatches the door, or denies entry. Entry is permitted when the token holds at least three data items: (1)
an ID certificate, which identifies the user, (2) a privilege certificate, which stores a clearance level, and (3) an
identification and authentication (I\&A) certificate, which assigns a fingerprint template. When the user first presents
their token the three certificates are read and cross-checked. If the token is valid, then a fingerprint is taken, which
if validated against the I\&A certificate, allows the door to be unlocked once the token is removed. An optional
authorisation certificate is also written upon successful authentication which allows the fingerprint check to be
skipped.

The TIS has a variety of other functions related to its
administration. Before use, a TIS must be enrolled, meaning it is
loaded with a public key chain and certificate, which are needed to
check token certificates. Moreover, the TIS stores audit data which
can be used to check previously occurred entries. The TIS therefore
also has a keyboard, floppy drive, and screen to configure
it. Administrators are granted access to these functions. The TIS also
has an alarm which will sound if the door is left open for too long.

Our objective is to construct an assurance case that argues that the \ac{tis} fulfils its security properties and
complies to the \ac{cc} standard~\cite{CC2017-CommonCriteriaInformationPt1}. \ac{cc} supports a vendor in delivering a
system compliant to a \emph{security level} while a certification authority confirms compliance and further qualities.
The standard defines \emph{seven} \acp{eal}, each a collection of \acp{sfr} and \acp{sar} the system would have to meet.
\acp{fm} are strongly recommended for \ac{eal} 5 and above.
Now, one can either (a) use a pre-specified \emph{set of general security
  properties},
or (b) develop an \emph{application-specific set} with the potential of
additional effort due to requirements analysis.
The security of the \ac{tis} is assured according to (b) by demonstrating six \acp{sfr}~\cite{TIS-SecurityProperties},
of which the first four are shown here and detailed below in~\autoref{sec:model}:
{\small
\label{props:TIS}
\begin{description}%[label={},leftmargin=*]
\item[SFR1] \label{prop:SFR1} If the latch is unlocked by TIS, then
  TIS must be in possession of either a User Token or an Admin
  Token. The User Token must either have a valid Authorisation
  Certificate, or must have valid ID, Privilege, and I\&A
  Certificates, together with a template that allowed TIS to
  successfully validate the user's fingerprint. Or, if the User
  Token does not meet this, the Admin Token must have a valid
  Authorisation Certificate, with role of ``guard''.
\item[SFR2] If the latch is unlocked automatically by TIS, then the current time must be close to being within the
  allowed entry period defined for the User requesting access.
\item[SFR3] An alarm will be raised whenever the door/latch is
  insecure.
\item[SFR4] No audit data is lost without an audit alarm being
  raised.
\end{description}}
The \emph{pioneering work} on the assurance of the \ac{tis} according
to option~(b) was carried out by Praxis High Integrity Systems and
SPRE~Inc.~\cite{Barnes2006-EngineeringTokeneerenclave}.  Barnes et
al.~performed security analysis, specification using Z, implementation
in SPARK, and verification and test of the security properties.  After
independent assessment \ac{eal} 5 was achieved.  This way, Tokeneer
became a successful example of the use of \acp{fm} in assuring a
system against \ac{cc}.

%%%%%%%%%%%%%%%%%%%%%%%%%%%%%%%
\section{\isacm}
\label{subsec:isacm}
%%%%%%%%%%%%%%%%%%%%%%%%%%%%%%%
In the following we encode SACM in \idof as an ontology, and then use it to provide a concrete syntax for our assurance
case language. Our embedding implements assurance cases as meta-logical entities. We are not embedding assurance
arguments in the HOL logic, as this would prevent the expression of informal reasoning and explanation. Rather, SACM is
implemented as a datatype in SML, meaning that we can refer to entities like types, terms, and
theorems as objects. Thus, certain claims can contain formal expressions, but others may have unstructured natural
language. Thus, we faithfully represent the inherently semi-formal nature of assurance cases.

We focus on the \inlineisar+ArgumentationPackage+\footnote{We model all parts of argumentation, base, artifact and
  terminology packages in \idof, but omit details about these for space reasons.} from Figure~\ref{fig:sacm-arg}, as this is
most relevant for the TIS argument we develop in \S\ref{sec:tokassure}. Different types of evidences and context,
modelled by the \inlineisar+ArtifactPackage+, can support a claim. The class \inlineisar+ArgumentAsset+ is represented
in \idof as follows:
\begin{isar}[numbers=none, backgroundcolor=\color{black!10}, frame=lines]
(*@\textcolor{Blue}{doc\_class}@*) ArgumentAsset = ArgumentationElement +
  content_assoc:: MultiLangString 
\end{isar}
Here, \inlineisar+ArgumentationElement+ is a base class which \inlineisar+ArgumentAsset+ inherits from, but is not
discussed further. The \inlineisar+content_assoc::+ \inlineisar+MultiLangString+ is an attribute modelling an
association between \inlineisar+ArgumentAsset+ and \inlineisar+MultiLangString+ from the \inlineisar+BasePackage+. It
allows classes inherited from \inlineisar+ArgumentAsset+ to include content expressed in multiple languages, and also
structured expressions, using the \inlineisar+TerminologyPackage+. Our implementation of \inlineisar+MultiLangString+
allows us to embed a variety of informal and formal content utilising the Isabelle term language.

The class \inlineisar+ArgumentAsset+ is inherited by three classes: (1) \inlineisar+Assertion+, which is a unified type for
claims and their relationships; (2) \inlineisar+ArgumentReasoning+, which is used to explicate the argumentation
strategy being employed; and (3) \inlineisar+ArtifactReference+, that evidences a claim with an artifact.

In \idof, \inlineisar+ArgumentAsset+ is inherited as follows:
\begin{isar}[numbers=none, backgroundcolor=\color{black!10}, frame=lines]
(*@\textcolor{Blue}{datatype}@*) assertionDeclarations_t = 
  Asserted|Axiomatic|Defeated|Assumed|NeedsSupport
(*@\textcolor{Blue}{doc\_class}@*) Assertion = ArgumentAsset +  
  assertionDeclaration::assertionDeclarations_t          
(*@\textcolor{Blue}{doc\_class}@*) ArgumentReasoning  = ArgumentAsset + 
  structure_assoc::"ArgumentPackage option"
(*@\textcolor{Blue}{doc\_class}@*) ArtifactReference = ArgumentAsset +
  referencedArtifactElement_assoc::"ArtifactElement set"
\end{isar}
Here, \inlineisar+assertionDeclarations_t+ is an \ihol enumeration type, \inlineisar+set+ is the set type, and
\inlineisar+option+ is the optional type. The attribute \inlineisar+assertionDeclaration+ is of type
\inlineisar+assertionDeclarations_t+, which specifies the status of assertions. The attribute
\inlineisar+structure_assoc+ is an association to the class \inlineisar+ArgumentPackage+, which is not discussed
here. Finally, the attribute \inlineisar+referencedArtifactElement_assoc+ is an association to the
\inlineisar+ArtifactPackage+ allowing claims to reference artifacts.

The class \inlineisar+Claim+ is a leaf child class and inherits from the class
\inlineisar+Assertion+.  This means that an instance of \inlineisar+Claim+ has a \inlineisar+gid+, a
\inlineisar+MultiLangString+ description, and can be \inlineisar+Axiomatic+, \inlineisar+Asserted+, etc. The other child
class for \inlineisar+Assertion+ is:
\begin{isar}[numbers=none, backgroundcolor=\color{black!10}, frame=lines]
(*@\textcolor{Blue}{doc\_class}@*) AssertedRelationship = Assertion +
  isCounter::bool
  reasoning_assoc:: "ArgumentReasoning option"
\end{isar}
Here, \inlineisar+isCounter+ specifies whether the target of the relation is refuted by the source, and
\inlineisar+reasoning_assoc+ is an association to \inlineisar+ArgumentReasoning+, to elaborate the
strategy. \inlineisar+AssertedRelationship+ models the relationships between elements of type
\inlineisar+ArgumentAsset+. In addition to the inherited attributes from parent classes, the relationship classes have
the attributes \inlineisar+source+ and \inlineisar+target+. The source attribute carries the supporting elements and the
target carries the supported elements.

From the SACM ontology, we create a number of Isabelle commands that create elements of the meta-model.  Our approach
gives a concrete syntax for SACM in terms of Isabelle commands as follows. Instances of concrete leaf child classes in
the metamodel have concrete syntax consisting of an Isabelle top-level command. Instances of attributes of a leaf child
class, including the inherited ones, have a concrete syntax represented by an Isabelle (\textcolor{OliveGreen}{green})
subcommand. Instances of associations between leaf child classes have a concrete syntax represented by an Isabelle
subcommand. The command has the name of the represented association and has an input with the type of the association of
the underlying instance. A selection of the commands for SACM is shown below.

\begin{isar}[numbers=none, backgroundcolor=\color{black!10}, frame=lines]
(*@\textcolor{Blue}{CLAIM}@*) <gid> (*@\textcolor{OliveGreen}{CONTENT}@*) <MultiLangString>
(*@\textcolor{Blue}{ASSERTED\_INFERENCE}@*) <gid> (*@\textcolor{OliveGreen}{SOURCE}@*) <gid>* (*@\textcolor{OliveGreen}{TARGET}@*) <gid>*
(*@\textcolor{Blue}{ASSERTED\_CONTEXT}@*) <gid> (*@\textcolor{OliveGreen}{SOURCE}@*) <gid>* (*@\textcolor{OliveGreen}{TARGET}@*) <gid>*
(*@\textcolor{Blue}{ASSERTED\_EVIDENCE}@*) <gid> (*@\textcolor{OliveGreen}{SOURCE}@*) <gid>* (*@\textcolor{OliveGreen}{TARGET}@*) <gid>*
(*@\textcolor{Blue}{ARTIFACT}@*) <gid> (*@\textcolor{OliveGreen}{VERSION}@*) <text> (*@\textcolor{OliveGreen}{DATE}@*) <text> (*@\textcolor{OliveGreen}{CONTENT}@*) <MultiLangString>
\end{isar}

\noindent\textcolor{Blue}{\inlineisar+CLAIM+} creates a new claim with an identifier (\inlineisar+gid+), and content described by
a \inlineisar+MultiLangString+. \textcolor{Blue}{\inlineisar+ASSERTED\_INFERENCE+} creates an inference between several
claims. Is has subcommands \textcolor{OliveGreen}{\inlineisar+SOURCE+} and \textcolor{OliveGreen}{\inlineisar+TARGET+}
that are both lists of elements. The command ensures that the cited claims exist, otherwise an error message is
issued. \textcolor{Blue}{\inlineisar+ASSERTED\_CONTEXT+} similarly asserts that an entity should be treated as context
for another, and \textcolor{Blue}{\inlineisar+ASSERTED\_EVIDENCE+} associates evidence with a claim. The
\textcolor{Blue}{\inlineisar+ARTIFACT+} command creates an evidential artifact, with description, date, and content.

\begin{wrapfigure}{r}{0.5\linewidth}
  \vspace{-5.5ex}

  \centering
  \includegraphics[width=6cm]{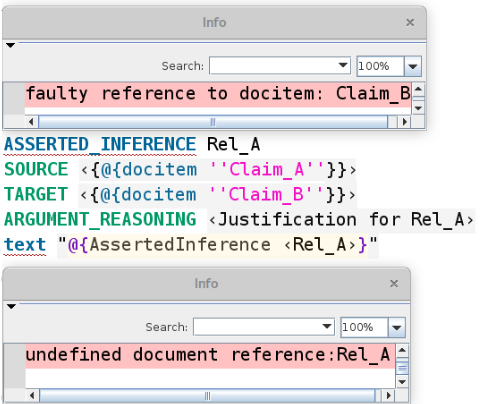}

  \vspace{-2ex}
  
  \caption{Relations in \isacm}
  
  \label{fig:ClaimA_DSL}

  \vspace{-4ex}
\end{wrapfigure}

Each command also has an associated antiquotation, which can be used to reference the entity type in a claim
string. This is illustrated in Figure~\ref{fig:ClaimA_DSL}, which shows the interactive nature of the assurance case
language. It represents an inferential link between a strategy and a justification (cf. Figure~\ref{fig:gsn}). An
asserted inference called \inlineisar+Rel_A+ has been created that attempts to link existing claims \inlineisar+Claim_A+
and \inlineisar+Claim_B+. However, \inlineisar+Claim_B+ does not exist, and so the error message at the top of the
screenshot is issued. A textual element is then created which references \inlineisar+Rel_A+ using the antiquotation
class \inlineisar+@{AssertedInference ...}+. This also leads to an error, shown at the bottom, since \inlineisar+Rel_A+
does not exist.

We have now developed our interactive assurance case tool. In the next section we begin to consider assurance of the
Tokeener system, first considering formal verification of the security properties.

%%%%%%%%%%%%%%%%%%%%%%%%%%%%%%%%%%%%%%%%%%
\section{Modelling and Verification of Tokeneer}
%%%%%%%%%%%%%%%%%%%%%%%%%%%%%%%%%%%%%%%%%%
\label{sec:model}
In this section we formally model the TIS in Isabelle/UTP~\cite{Foster19a-IsabelleUTP}, in order to provide evidence for
the assurance case. In~\cite{TIS-SecurityProperties}, the six SFRs are argued semi-formally, but here we provide a
formal proof. We focus on the formalisation and verification of the user entry part of SFR1, and describe the elements
necessary for this.

The TIS behaviour, formalised by Praxis in the Z notation~\cite{TIS-FormalSpec}, uses an elaborate state space and a
collection of relational operations. The state is bipartite, consisting of (1) the digital state of the TIS and (2) the
variables shared with the real world, which are monitored or controlled, respectively. The TIS monitors the time,
enclave door, fingerprint reader, token reader, and several peripherals. It controls the door latch, an alarm, a
display, and a screen.

The specification describes a state transition system, illustrated in Figure~\ref{fig:tis-states}
(cf. \cite[page~43]{TIS-FormalSpec}), where each transition corresponds to an operation. Several operations are ommited
due to space constraints. Following enrolment, the TIS becomes \textsf{quiescent} (awaiting
interaction). \textsf{ReadUserToken} triggers if the token is presented, and reads its contents. Assuming a valid token,
the TIS determines whether a fingerprint is necessary, and then triggers either \textsf{BioCheckRequired} or
\textsf{BioCheckNotRequired}. If required, the TIS then reads a fingerprint (\textsf{ReadFingerOK}), validates it
(\textsf{ValidateFingerOK}), and finally writes an authorisation certificate to the token
(\textsf{WriteUserTokenOK}). If the access credentials are available (\textsf{waitingEntry}), then a final check is
performed (\textsf{EntryOK}), and once the user removes their token (\textsf{waitingRemoveTokenSuccess}), the door is
unlocked (\textsf{UnlockDoor}).

\begin{figure}[t]
  \centering\includegraphics[width=\linewidth]{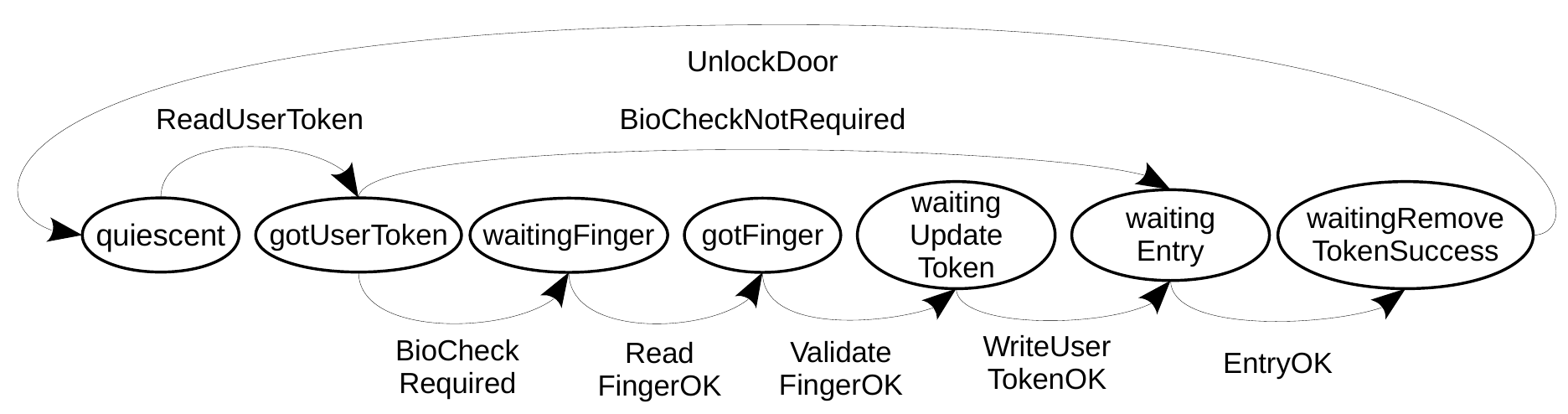}

  \vspace{-2ex}

  \caption{TIS Main States}
  \label{fig:tis-states}

  \vspace{-4ex}
\end{figure}

We mechanise the TIS model by first creating hierarchical state space types, with invariants adapted from the Z
specification~\cite{TIS-FormalSpec}. We define the operations using Dijkstra's guarded command
language~\cite{Dijkstra75} (GCL) rather than the Z schemas directly, as GCL is easier to reason about and provides
similar expressivity. Moreover, GCL is given a denotational semantics in UTP's alphabetised relational
calculus~\cite{Hoare&98}, and so it is possible to prove equivalence with the corresponding Z operations. We use a
variant of GCL that broadly follows the following syntax:
$$\prog ::= \ckey{skip} | \ckey{abort} | \prog \relsemi \prog | \expr \longrightarrow \prog | \prog \intchoice \prog | \var := \expr | \uframe{\var}{\prog} $$
Here, $\prog$ is a program, $\expr$ is an expression or predicate, and $\var$ is a variable. The language provides the
usual syntax for sequential composition, guarded commands, non-deterministic choice, and assignment. We also adopt a
framing operator $\uframe{a}{P}$ which states that $P$ can refer only to variables in the namespace $a$, and all other
variables remain unchanged~\cite{Foster16a,Foster19a-IsabelleUTP}.

We now introduce the TIS state space, on which the state machine will act.
\begin{align*}
% \textit{DOOR} &\defs \mv{open} | \mv{closed} \quad \quad \textit{LATCH} \defs \mv{unlocked} | \mv{locked} \\
% \textit{TOKENTRY} &\defs \mv{noT} | \mv{badT} | \mv{goodT}~\textit{Token} \quad \textit{PRESENCE} \defs \mv{present} | \mv{absent} \\
IDStation &\defs \left[
\begin{array}{l}
  currentUserToken : TOKENTRY, currentTime : TIME, \\ 
  userTokenPresence : PRESENCE, status : STATUS, \\ 
%  enclaveStatus : \textit{ENCLAVESTATUS}, \\
%  currentDisplay : \textit{DISPLAYMESSAGE},  \\
  issuerKey : \textit{USER} \pfun \textit{KEYPART}, \cdots
\end{array}
\right] \\[.5ex]
Controlled &\defs \left[
\begin{array}{l}
  latch : LATCH, alarm : ALARM, \cdots
%  display : DISPLAYMESSAGE, screen : Screen
\end{array}
\right] \\[.5ex]
Monitored &\defs \left[
\begin{array}{l}
  now : TIME, finger : \textit{FINGERPRINTTRY}, \\ 
  userToken : TOKENTRY, \cdots \\
%  floppy : FLOPPY, keyboard : KEYBOARD
\end{array} 
\right] \\[.5ex]
RealWorld &\defs \left[ mon : Monitored, ctrl : Controlled \right] \\[.5ex]
SystemState &\defs [ rw : RealWorld, tis : IDStation ]
\end{align*}
We define five state space types that describe the TIS state, the controlled variables , monitored variables,
real-world, and the entire system, respectively. The controlled variables include the physical latch, the alarm, the
display, and the screen. The monitored variables correspond to time ($now$), the door ($door$), the fingerprint reader
($finger$), the tokens, and the peripherals. \textit{RealWorld} combines the physical variables, and
\textit{SystemState} composes the physical world and TIS.

Variable \textit{currentUserToken} represents the last token presented to the TIS, and \textit{userTokenPresence}
indicates whether a token is currently presented. The variable \textit{status} is used to record the state the TIS is
in, and can take the values indicated in the state bubbles of Figure~\ref{fig:tis-states}. Variable \textit{issuerKey}
is a partial function representing the public key chain, which is needed to authorise user entry.

We now specify a selection of the operations over this state space:
\begin{align*}
  BioCheckRequired &\defs 
  \begin{array}{l}
  \left(\begin{array}{l}
    status = \mv{gotUserToken} \land userTokenPresence = \mv{present} \\
    \land UserTokenOK \land (\neg UserTokenWithOKAuthCert)
  \end{array}\right) \\[1ex]
  \longrightarrow status := \mv{waitingFinger}\!\relsemi currentDisplay := \mv{insertFinger}
  \end{array} \\
  ReadFingerOK &\defs
  \begin{array}{l}
  \left(\begin{array}{l}
    status = \mv{waitingFinger} \land fingerPresence = \mv{present} \\
    \land userTokenPresence = \mv{present}
  \end{array}\right) \\
  \longrightarrow status := \mv{gotFinger} \relsemi currentDisplay := \mv{wait}
  \end{array} \\
  UnlockDoorOK &\defs
  \begin{array}{l}
  \left(\begin{array}{l}
    status = \mv{waitingRemoveTokenSuccess} \\
    \land userTokenPresence = \mv{absent}
  \end{array}\right) \\
  \longrightarrow \begin{array}{l} UnlockDoor \relsemi status := \mv{quiescent} \relsemi \\ currentDisplay := \mv{doorUnlocked} \end{array}
  \end{array}
\end{align*}
Each operation is guarded by execution conditions and consists of several assignments. \textit{BioCheckRequired}
requires that the current state is $\mv{gotUserToken}$, the user token is $\mv{present}$, and sufficient for entry
($UserTokenOK$), but there is no authorisation certificate ($\neg UserTokenWithOKAuthCert$).  The latter two predicates
essentially require that (1) the three certificates can be verified against the public key store, and (2) additionally
there is a valid authorisation certificate present. Their definitions can be found
elsewhere~\cite{TIS-FormalSpec}. \textit{BioCheckRequired} updates the state to $\mv{waitingFinger}$ and the display
with an instruction to provide a fingerprint. \textit{UnlockDoorOK} requires that the current state is
$\mv{waitingRemoveTokenSuccess}$, and the token has been removed. It unlocks the door, using the elided operation
\textit{UnlockDoor}, returns the status to $\mv{quiescent}$, and updates the display.

These operations act only on the TIS state space. During their execution monitored variables can also change, to reflect
real-world updates. Mostly these changes are arbitrary, with the exception that time must increase monotonically. We
therefore promote the operations to \textit{SystemState} with the following schema.
$$UEC(Op) \defs \uframe{tis}{Op} \relsemi \uframe{rw}{\lns{mon}{now} \le \lns{mon}{now'} \land ctrl' = ctrl}$$
In Z, this functionality is provided by schema \textit{UserEntryContext}~\cite{TIS-FormalSpec}, from which we derive the
name \textit{UEC}. It promotes $Op$ to act on $tis$, and composes this with a relation that specifies changes to the
real-world variables ($rw$). We specify this as a UTP relational predicate. The behaviour of all monitored variables
other than $now$ is arbitrary, and all controlled variables are unchanged. Then, we promote each operation, for example
$TISReadTokenOK \defs UEC(ReadTokenOK)$. The overall behaviour of the entry operations is given below:
$$
 TISUserEntryOp \defs \left(
 \begin{array}{l}
   TISReadUserToken \intchoice TISValidateUserToken \\
   \intchoice TISReadFinger \intchoice TISValidateFinger \\
   \intchoice TISUnlockDoor \intchoice TISCompleteFailedAccess \intchoice \cdots
 \end{array}\right)
$$
In each iteration of the state machine, we non-deterministically select an enabled operation and execute it. We also
update the controlled variables, which is done by composition with the following update operation.
\begin{align*}
  TISUpdate \defs~& \uframe{rw}{\lns{mon}{now} \le \lns{mon}{now'}} \relsemi \lns{rw}{\lns{ctrl}{latch}} := \lns{tis}{currentLatch} \relsemi  \\
            & \lns{rw}{\lns{ctrl}{display}} := \lns{tis}{currentDisplay}
\end{align*}
We also formalise the TIS state invariants necessary to prove SFR1:
  \begin{align*}
    Inv_1 &\defs 
      \begin{array}{l}
      status \in 
            \left\{\begin{array}{l}
                     \mv{gotFinger}, \mv{waitingFinger}, \mv{waitingUpdateToken} \\
                     \mv{waitingEntry}, \mv{waitingUpdateTokenSuccess}
                   \end{array}\right\} \\
      \implies (UserTokenWithOKAuthCert \lor UserTokenOK)
      \end{array} \\
    Inv_2 &\defs 
      \begin{array}{l}
      status \in \{\mv{waitingEntry}, \mv{waitingUpdateTokenSuccess}\} \\ 
        \implies (UserTokenWithOKAuthCert \lor FingerOK)
      \end{array} \\
    \textit{TIS-inv} &\defs Inv_1 \land Inv_2 \land \cdots 
  \end{align*}
%\end{definition}
%
\noindent $Inv_1$ states that whenever the TIS is in a state beyond $\mv{gotUserToken}$, then either a valid
authorisation certificate is present, or else the user token is valid. $Inv_2$ states that whenever in state
$\mv{waitingEntry}$ or $\mv{waitingUpdateTokenSuccess}$, then either an authorisation certificate or a valid finger
print is present. We elide the additional invariants that deal with the alarm and audit data~\cite{TIS-FormalSpec}.

Next, we show that each operation preserves \textit{TIS-inv} using Hoare logic.

\begin{theorem} $\hoaretriple{\textit{TIS-inv}}{TISUserEntryOp}{\textit{TIS-inv}}$
\end{theorem}
\begin{proof} Automatic, by application of Isabelle/UTP proof tactic \textsf{hoare-auto}. \end{proof}

\noindent This theorem shows that the state machine never violates the invariants, and we can assume they hold to
satisfy any requirements. We use this to formalise and prove SFR1, which requires that we determine all states under
which the latch will be unlocked. We can determine these states by application of the weakest precondition
calculus~\cite{Dijkstra75}. Specifically, we characterise the weakest precondition under which execution of
\textit{TISUserEntryOp} followed by \textit{TISUpdate} leads to a state satisfying
$\lns{rw}{\lns{ctrl}{latch}} = \mv{unlocked}$. We formalise this in the theorem below.

\begin{theorem}[FSFR1 is satisfied]
  \begin{align*}
    & \left(\begin{array}{l}\textit{TIS-inv} \land \lns{tis}{currentLatch} = \mv{locked} \\
      \land (TISEntryOp \relsemi TISUpdate)\mathop{\,\ckey{wp}\,}(\lns{rw}{\lns{ctrl}{latch}} = \mv{unlocked})
      \end{array} \right) \\
    & \quad \implies ((\textit{UserTokenOK} \land \textit{FingerOK}) \lor \textit{UserTokenWithOKAuthCert})
  \end{align*}
\end{theorem}
\begin{proof}
  Automatic, by application of weakest precondition and relational calculi.
\end{proof}

We calculate the weakest precondition, and conjoin this with $\textit{TIS-Inv}$, which always holds, and the predicate
$\lns{tis}{currentLatch} = \mv{locked}$ to capture behaviours when the latch was initially locked. We show that this
composite precondition implies that either a valid user token and fingerprint were present, or else a valid
authorisation certificate. We have therefore now verified a formalisation of SFR1. In the next section we place this in
the context of an assurance argument.

%%%%%%%%%%%%%%%%%%%%%%%%%%%%%%%%%%%%%%%%%%%%%%%%%
\section{Mechanising the Tokeener Assurance Case}
%%%%%%%%%%%%%%%%%%%%%%%%%%%%%%%%%%%%%%%%%%%%%%%%%
\label{sec:tokassure}
\begin{wrapfigure}{r}{0.55\linewidth}
  \vspace{-5.5ex}

  \centering
  \includegraphics[width=\linewidth]{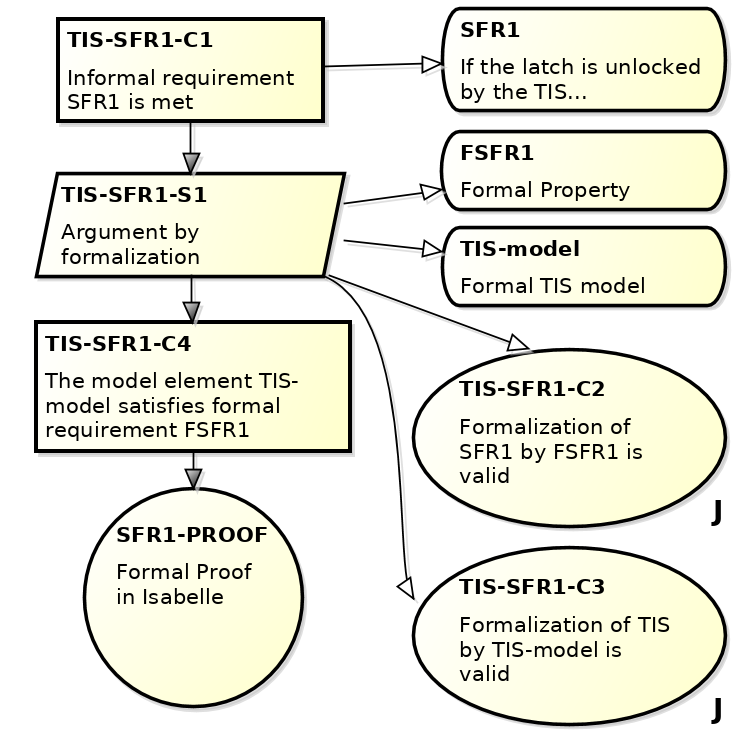}

  \vspace{-1.5ex}
  
  \caption{TIS Claim Formalization}
  
  \label{fig:formalclaimpattern}

  \vspace{-4ex}

\end{wrapfigure}  
In the following, we mechanize an assurance argument with the claim that \ac{tis} satisfies \textbf{SFR1}. The assurance
case fragment is shown in \autoref{fig:formalclaimpattern}, which is inspired by the formalisation
pattern~\cite{Denney2018}. The latter shows how results from a formal method can be employed in an assurance case. This
is contingent on the validation of both the formal model and the formal requirement. Consequently, the formalisation
pattern breaks down a requirement into 3 claims stating that (1) the formal model is validated, (2) the formal
requirement correctly characterises the informal requirement, and (3) the formal requirement is satisfied by the formal
model. The former two claims will usually have an informal process argument.

The argument in \autoref{fig:formalclaimpattern} justifies the link between the informal claim ``\ac{tis} satisfies
\textbf{SFR1}'', which is in natural language, and the formal theorem \inlineisar+FSRF1+ from \S\ref{sec:model}, which
is expressed in \ac{hol}. The top-level claim, \textbf{TIS-SFR1-C1}, states that SFR1 is satisfied, which it references
as a contextual element. This claim is decomposed by the use of the formalization strategy, \textbf{TIS-SFR1-S1}, which
has the formal property (FSRF1) and TIS model from \S\ref{sec:model} as context. The satisfaction of the formal claim is
expressed by \textbf{TIS-SFR1-C4}, and evidenced by \textbf{SFR1-PROOF}, which is the formal proof. The validation
claims are encoded as justifications \textbf{TIS-SFR1-C3} and \textbf{TIS-SFR1-C4}, which are not elaborated.

\autoref{fig:Claim_C} contains a mechanised version of the same argument in \isacm. The structure is slightly different
from the GSN diagram since justifications are particular kind of claims in SACM. The five claims are specified using the
\textcolor{Blue}{\inlineisar+CLAIM+} command, with a name and content associated. The text in these claims integrate
hyper-linked semi-formal and formal content; for example \textbf{TIS-SFR1-C1} uses the antiquotation
\inlineisar+Expression+ to insert a formal link to the defined expression for SFR1. Similarly, \textbf{TIS-SFR1-C3}
contains a reference to the resource artifact \textbf{TIS}, which is a reference to the Tokeneer specification, and also
an Isabelle constant \textbf{TIS-model} which contains the formal TIS model.

\begin{figure}[t]
  \vspace{-1ex}

  \centering
  \includegraphics[width=\linewidth]{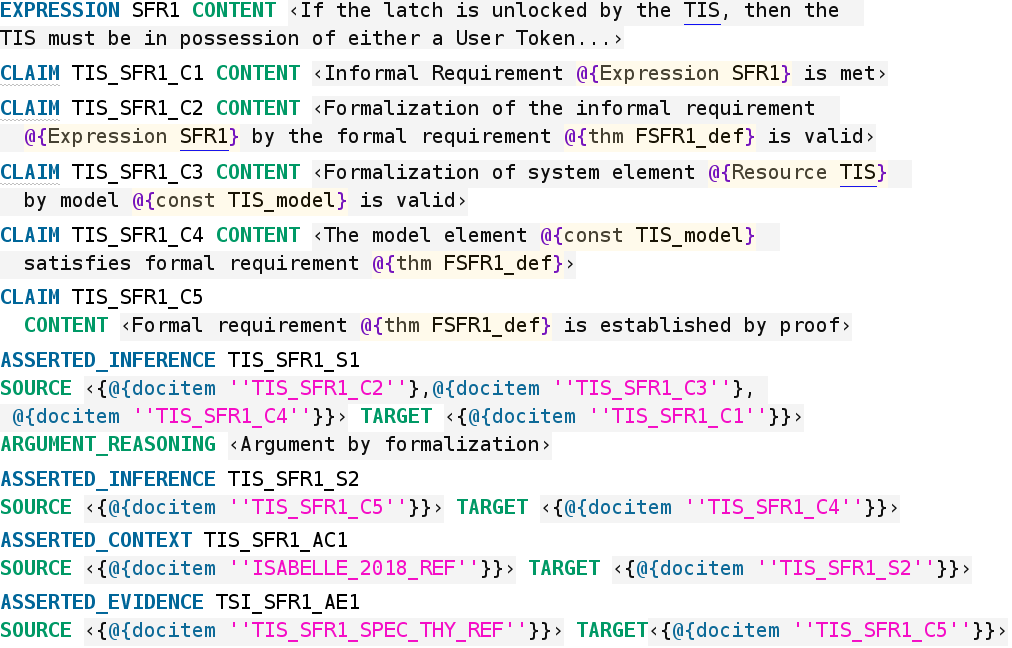}

  \vspace{-1.5ex}
  
  \caption{TIS argument: Claims and their relations in \isacm}
  
  \label{fig:Claim_C}

  \vspace{-4ex}
\end{figure}

The inferences between these claims are specified by several instances of the
\textcolor{Blue}{\inlineisar+ASSERTED_INFERENCE+} command, each of which links several premise claims to one or more
conclusions. \textbf{TIS-SFR-S1} shows that satisfaction of the informal requirement depends on the formal requirement,
and the two validation claims.

\autoref{fig:formalclaimpattern} does not elaborate further on the evidential artifacts required, for the verification,
as GSN does not support this. This is functionality which SACM supports through the artifact meta-model, which allows us
to record activities, participants, resources, and other assurance artifacts. \autoref{fig:Artifact_Rels} supplements
the argument in \autoref{fig:Claim_C} with the various evidential artifacts, and the relationships between them. For
example, the evidence supporting the claim \inlineisar+TIS_SFR1_C4+ is a reference to the artifact
\inlineisar+TIS_FSFR1_SPEC_THY+ which is the Isabelle theory containing the proof of the theorem stated by
\inlineisar+TIS_SFR1_C4+. \inlineisar+TIS_FSFR1_SPEC_THY+ refers via the artifact relationship
\inlineisar+TIS_SFR1_PROOF_ACTIVITY_REL+ to the context of the proof which is the artifact \inlineisar+Isabelle_IT+ and
to the author of the proof which is \inlineisar+Simon_Foster+. This illustrates how \isacm allows us to combine informal
artifacts and activities with the formal results they produce.

\begin{figure}[t]
  \vspace{-1ex}

  \centering
  \includegraphics[width=\linewidth]{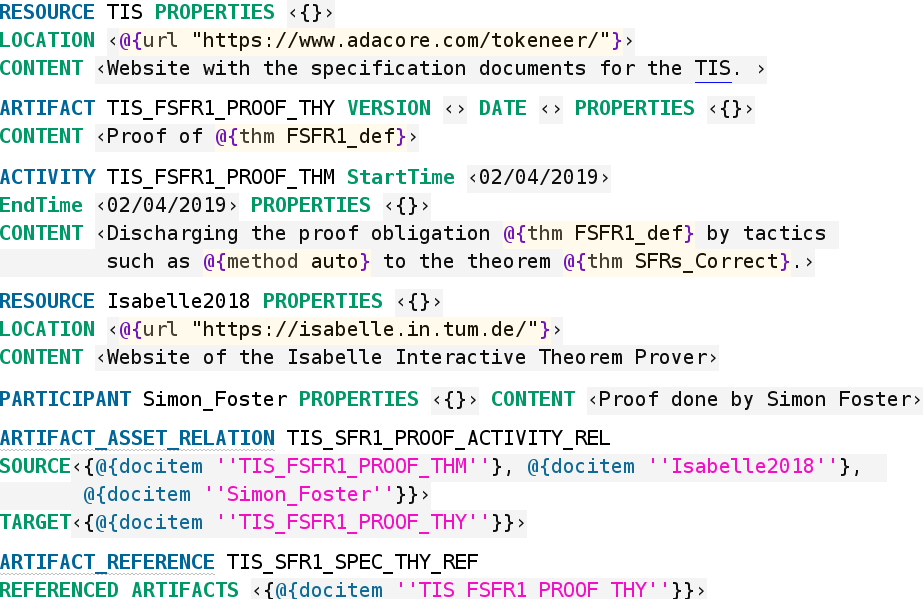}

  \vspace{-1.5ex}
  
  \caption{TIS argument: Artifacts and their relations in \isacm}
  
  \label{fig:Artifact_Rels}

  \vspace{-5ex}
\end{figure}

%%%%%%%%%%%%%%%%%%%%%%
\section{Related Work}
\label{subsec:relatedWork}
%%%%%%%%%%%%%%%%%%%%%%
Woodcock et al.~\cite{Woodcock2010-TokeneerExperiments} highlight defects of the Tokeneer SPARK implementation, indicate
undischarged verification conditions, and perform robustness tests generated by the Alloy SAT
solver~\cite{Jackson2000-AlloyLightweightObject} model. Using De Bono's lateral thinking, these test cases go beyond the
anticipated operational envelope and stimulate anomalous behaviours of the Tokeneer implementation.  In shortening the
feedback cycle for verification and test engineers, interactive theorem proving can help using Woodcock's approach more
intensively.

Rivera et al.~\cite{Rivera2016-Undertakingtokeneerchallenge} present an Event-B model of the TIS~(the enrolment
operations of the admin are presented), verify this model, generate Java code from it using the Rodin tool, and test
this code by JUnit tests manually derived from the specification.  The tests validate the model in addition to the
Event-B invariants derived from the same specification. The tests aim to detect errors in Event-B model caused by
misunderstandings of the specification. Using Rodin, the authors verify the security
properties~(Section~\ref{props:TIS}) using Hoare triples.

We believe that our work is the first to put formal verification effort into the wider context of structured assurance
argumentation, in our case, a machine-checked security case using
Isabelle/SACM. 

Several works bring formality to assurance
cases~\cite{Denney2013,Cruanes2013,Rushby2014,Denney2018,Diskin2018}. AdvoCATE is a tool for the construction of
GSN-based safety cases~\cite{Denney2013,Denney2018}. It uses a formal foundation called argument structures, which
prescribe well-formedness checks for the graph structure, and allow instantiation of assurance case patterns. Our work
likewise ensures well-formedness, and additionally allows the embedding of formal content. Denney's formalisation
pattern~\cite{Denney2018} is an inspiration for our work.

Rushby shows how assurance arguments can be embedded into formal logic to overcome logical
fallacies~\cite{Rushby2014}. Our framework similarly allows reasoning using formal logic, but additionally allows us to
combine formal and informal artifacts. We were also inspired by the work on evidential toolbus~\cite{Cruanes2013}, which
allows the combination of evidence from several formal and semi-formal analysis tools. Isabelle similarly allows
integration of a variety of formal analysis tools~\cite{Wenzel2007}.

%%%%%%%%%%%%%%%%%%%%%
\section{Conclusion}
\label{sec:conclusion}
%%%%%%%%%%%%%%%%%%%%
We have presented \isacm, a framework for mechanised assurance cases. We showed how SACM is embedded into Isabelle as an
ontology, and provided an interactive assurance language that generates valid instances. We applied it to the production
of part of the Tokeneer security case, including verification of one of the security functional requirements, and
embedded these results into a mechanised assurance argument. Of a particular note, \isacm enforce the usage of the formal
ontological links which represent the relationships between the assurance arguments and their claims, a feature we
inherit from \idof. \isacm also combines features from \ihol, \idof and \sacm which results in a framework that allows
the integration of formal methods and argument-based safety assurance cases.

In future work, we will consider the integration of assurance case pattern execution~\cite{Denney2018} into our
framework, which facilitate their production. We will also complete the mechanisation of the TIS security case,
including verification of the other five SFRs. In parallel with this, we are developing our verification framework,
Isabelle/UTP~\cite{Foster19a-IsabelleUTP} to support a variety of notations used in software engineering. We recently
demonstrated formal verification facilities for a statechart-like notation~\cite{Foster18b}, and are also working
towards tools to support hybrid dynamical languages~\cite{Foster16b} like Modelica and Simulink. Our overarching goal is a
comprehensive assurance framework supported by a variety of integrated formal methods in order to approach complex
certification tasks for cyber-physical systems and autonomous robots.

\bibliographystyle{splncs03}
\bibliography{FM2019}

\end{document}